\documentclass[a4paper, 10pt,psamsfonts]{amsart}

\usepackage{amsmath}
\usepackage{amsthm}
\usepackage{amssymb}
\usepackage{amscd}
\usepackage{amsfonts}
\usepackage{amsbsy}
\usepackage{setspace}
\usepackage{color}
\usepackage{graphicx}
\usepackage{calc}
\usepackage{subfigure}
\usepackage{multirow}

\newcommand{\R}{\ensuremath{\mathbb{R}}}
\newcommand{\C}{\ensuremath{\mathbb{C}}}

\newcommand{\X}{\ensuremath{\mathfrak{X}}}
\newcommand{\p}{\partial}

\newtheorem {theorem} {Theorem} 
\newtheorem {proposition} [theorem] {Proposition}
\newtheorem {corollary} [theorem] {Corollary}

\newcommand{\cF}{\mathcal{F}}

\title{Integrability of the Hide--Skeldon--Acheson dynamo}

\author{Adam Mahdi$^{1,2}$ and Claudia Valls$^{3}$}

\address{$^1$ Mathematics Department,
North Carolina State University, Raleigh, NC 27695-8205,
USA} \email{amandi@ncsu.edu}

\address{$^2$ Faculty of Applied Mathematics, AGH University of Science and Technology, al. Mickiewicza 30, 30-059 Krak\'ow, Poland}

\address{$^3$ Departamento de Matem\'atica, Instituto Superior T\'ecnico,
Av. Rovisco Pais 1049-001, Lisboa, Portugal}
\email{cvalls@math.ist.utl.pt}



\begin{document}

\maketitle

\begin{abstract}
In this work we consider the Hide-Skeldon-Acheson dynamo
model
\[
\dot x=x(y-1)-\beta z, \quad \dot y =\alpha(1-x^2)-\kappa y, \quad \dot z =x-\lambda
z,
\]
where $\alpha,\beta,\kappa$ and $\lambda$ are parameters.
 We contribute to the understanding of its global
dynamics, or more precisely, to the topological structure of its
orbits by studying the integrability problem. Provided $\alpha \ne
0$ we identify the values of the parameters of this model, for which
it admits a first integral. Also, as corollary of our main results we get that  for $\alpha, \beta, \kappa \ne 0$ the dynamo model does not admit
a polynomial, rational or Darboux first integral.
\end{abstract}

\section{Introduction}
In 1996 Hide, Skeldon and Acheson \cite{HidSkeAch96} proposed a
model for self--exciting dynamo action in which a Faraday disk and
coil are arranged in series with either a capacitor or a motor. The
governing equations for these dynamo models are
\begin{equation}\label{Nose}
\dot x=x(y-1)-\beta z, \quad \dot y =\alpha(1-x^2)-\kappa y,
\quad \dot z =x-\lambda z,
\end{equation}
where $\alpha,\beta,\kappa,\lambda$ are real parameters. In what follows
those models will be called
{\it HSA dynamo}.
In \eqref{Nose} $x=x(t)$ is the current flowing
through the dynamo, $y=y(t)$ is the
angular rotation rate of the disk and $z=z(t)$ is the angular
rotation rate of the motor (or the charge in the capacitor).  The
system contains four parameters $\alpha, \beta, \kappa$ and
$\lambda$. The first one $\alpha$ is proportional to the steady
applied mechanical couple driving the disk into rotation, and
$\beta^{-1}$ measures the moment of inertia of the armature of the
motor, and $\kappa$ and $\lambda$ are the coefficients of friction
in the disk and the motor, respectively. For the derivations of
those equations see \cite{HidSkeAch96}. Since the HSA dynamo was derived for the first time, a number of its
features were revealed. For example in \cite{Hid97a, Hid97b} the
authors have extended the HSA dynamo including the effects of a
nonlinear motor, an external battery and magnetic field, the
coupling of two or more identical dynamos together. In subsequent
works \cite{MorHidSow98, GolMorHid00, Mor01} and \cite{Mor02} the
authors analysed the influence on the HSA model of nonlinearity of
the series motor, the presence of a series battery and/or an ambient
magnetic field. Bifurcation transition diagrams have been presented
in \cite{Mor07}, where the author also identifies unstable periodic
orbits pertaining to some cases.  The analysis towards identifying
the underlying chaotic attractor was done in \cite{Mor08}, where the
author used some of the unstable periodic orbits to identify a
possible template for the chaotic attractor, using ideas from
topology.

The question whether a differential model admits a first integral (for the precise definition see Section \ref{Sec:Prelim}) is of fundamental importance.
One reason is that the first integrals give conservation laws for the model and that enables one to lower its dimension. Moreover, knowing a sufficient number of first integrals
allows to solve the system explicitly. Finally the existence or non-existence \cite{LliMahVal, MahVal:Magnetic, MahVal:Suslov} of first integrals for a given model measures, in a sense, the complexity of the set of its orbits. There are two equally difficult problems. One is to prove that a given system is chaotic; or to prove that it is not. In this paper we study the integrability problem for the HSA dynamo. We contribute to the understanding of the complexity, or more precisely to the topological structure of the orbits of HSA dynamo by studying the integrability problem for this model depending on its
four parameters $\alpha, \beta, \gamma$ and $\kappa$.  For proving our main
results we shall use the information about invariant algebraic 
surfaces of this system, which is the basis of the so called Darboux
theory of integrability, for more details see Section
\ref{Sec:Prelim}. We also note that for $\kappa=\lambda=0$ and $\beta=1$,
the HSA dynamo is equivalent to the Nos{\'e}--Hoover equation (cf.
\cite{SwiWag08} and \cite{Chi-Review}), for which the first
integrals have been studied in \cite{MahVal10:Nose}.

In the following theorem we give explicit formulas for the first
integrals of the HSA dynamo in the case that $\alpha\ne 0$.
\begin{theorem}\label{thm:integrals}
Assume that $\alpha\ne0$ and define the following functions:
\[
\begin{aligned}
\cF_1&= \alpha \ln x -\alpha x^2/2 -y^2/2 +y;\\
\cF_2&= z v(x)-\int v(x) w(x)\,dx;\\
\cF_{3}&=\kappa\Big[\log\frac{2(\kappa+y-1+\mathcal{T})}{-x}\Big]-\mathcal{T};\\
\cF_{4}&=\frac{y+\kappa-1 -x\sqrt{\kappa}
\sqrt{\kappa-1}}{y+\kappa-1+ x\sqrt{\kappa}\sqrt{\kappa-1}}
\exp\big\{2z \sqrt{\kappa} \sqrt{\kappa-1}\big\};
\end{aligned}
\]
where
\[
w(x)= \Big[1-2 (\cF_3 +\alpha \frac{x^2}2-\alpha \ln x
)\Big]^{-1/2},\quad v(x)=\exp\Big[\lambda\int xw(x)\,dx\Big],
\]
\[
\mathcal{T}=\sqrt{-\kappa^2(x^2\!-\!1)\!+\!(y\!-\!1)^2\!+\!\kappa(x^2\!+\!2y\!-\!2)}.
\]
Then HSA dynamo admits the following first integral:
\begin{enumerate}
\item[(a)] if $\beta=\kappa=0$ and $\lambda\in\R$, then $\cF_1$ and
$\cF_2$ are the first integrals;
\item[(b)] if $\beta=0$, $\alpha=-\kappa(\kappa-1)$ and $\kappa\ne
0$, then $\cF_{3}$ is a first integral. Additionally, if $\lambda=0$ then
$\cF_{4}$ is also a first integral.
\end{enumerate}
\end{theorem}

It is a straightforward computation to check that
$\mathcal{F}_1, \mathcal{F}_2, \mathcal{F}_3$ and $\mathcal{F}_4$ are the first integrals of
system \eqref{Nose} for the corresponding values of the parameters,
thus the proof of Theorem \ref{thm:integrals} will be omitted. One
should not assume that it is trivial or easy to find those functions even
though, once we know them, it is not difficult a check that they are first integrals. To date there are no
general methods that would allow to decide whether a given system of
differential equations is integrable or would give a way to
calculate its first integrals.  In the particular case of polynomial
differential systems one of the best tools to approach this problem  is the so-called Darboux theory of integrability. See Section 2 where we briefly explain how to use invariant
algebraic surfaces (in general hypersurfaces) and exponential
factors to construct a first integral.

Now we present other main results of the paper. The only values of the parameters of the HSA dynamo, for which it admits one or two Darboux first integrals (which include polynomial one) are given in Theorem 1. The following two theorems address the integrability problem for the HSA dynamo for the values of $\alpha, \beta, \kappa$ and $\lambda$ not considered in Theorem 1. Thus they can be viewed as a non-integrability results. We consider two cases $\beta=0$ (Theorem 1) and $\beta \neq 0$ (Theorem 2).

\begin{theorem}\label{thm.main.4}
The following statements hold for the HSA dynamo equation with
$\alpha \ne 0$, $\beta=0$, $\kappa \ne 0$,
$\alpha \ne -\kappa(\kappa-1)$ and $\lambda \in \R$:
\begin{enumerate}
\item[(a)] It does not admit any polynomial first integral.
\item[(b)] Its unique Darboux polynomial with nonzero cofactor is $x$.
\item[(c)] Its only exponential factor is $\exp(z)$
with the cofactor $x-\lambda z$.
\item[(d)] It is not Darboux integrable.
\end{enumerate}
\end{theorem}

Finally, we consider the case in which $\alpha, \beta \ne 0$ and $\kappa,
\lambda \in \R$.

\begin{theorem}\label{thm.main.5}
The following statements hold for the HSA dynamo equation with $\alpha,\beta
\ne 0$, and
 $\kappa,\lambda \in \R$:
\begin{enumerate}
\item[(a)] It does not admit any polynomial first integral.
\item[(b)] It does not admit any Darboux polynomial.
\item[(c)] Its only exponential factors are:
\begin{enumerate}
\item[(c.1)]
 $\exp(z)$
and $\exp(-x^2/2 +y/\alpha -y^2/(2 \alpha)-\beta z^2/2)$
with the cofactors $x-\lambda z$ and $y-1$, respectively if $\kappa=\lambda=0$,
\item[(c.2)]
 $\exp(z)$
with the cofactors $x-\lambda z$ in any other case.
\end{enumerate}
\item[(d)] It is not Darboux integrable.
\end{enumerate}
\end{theorem}
Clearly, if system \eqref{Nose} does not admit a Darboux first
integral, then it does not admit a rational first integral. Thus Theorem \ref{thm.main.4} and \ref{thm.main.5}
imply the following simple corollary.

\begin{corollary}
If $\alpha, \beta, \kappa \ne 0$, then the HSA dynamo does
not admit any polynomial, rational or Darboux first integral.
\end{corollary}

In Section \ref{Sec:Prelim} we present some preliminary results on
Darboux theory of integrability that will be used all through the
paper. In Section \ref{sec.5} we state some results on the HSA
dynamo equation for $\alpha \ne 0$. Finally, Theorems
\ref{thm.main.4} and \ref{thm.main.5} are proved in Sections
\ref{sec.6} and \ref{sec.7}, respectively.

\section{Preliminary results}\label{Sec:Prelim}
The associated vector field to \eqref{Nose} is
\begin{equation}\label{eq:vf}
\X=\big[x(y-1)-\beta z\big]\frac{\p}{\p x}+\big[\alpha(1-x^2)
-\kappa y \big] \frac{\p}{\p y}+\big[x-\lambda z\big]\frac{\p }{\p
z}.
\end{equation}
Let $U \subset \R^3$ be an open subset. We say that the
nonconstant function $H \colon U \to \R$ is a {\it first integral}
of the polynomial vector field \eqref{eq:vf} associated to system
\eqref{Nose}, if $H(x(t),y(t),z(t))=\text{constant}$ for all values
of $t$ for which the solution $(x(t),y(t),z(t))$ of $\X$ is defined
on $U$. Clearly $H$ is a first integral of $\X$ on $U$ if and only
if $\X H=0$ on $U$. When $H$ is a polynomial we say that it is a
\emph{polynomial first integral}. We will say that system
\eqref{Nose} is {\it integrable} if it admits a first integral. We
will also say that this system is {\it completely integrable} if it
admits two functionally independent first integrals. Since each
level curve of a first integral is invariant under the flow induced by
the system, it is clear that if this system is completely integrable,
then the intersection of its two first integrals determines an
invariant curve, which in turn gives information on the orbits of
the system.

In what follows we recall the basic notion from the Darboux theory
of integrability \cite{Lli04-Integ}. Let $h=h(x,y,z)\in\C[x,y,z]$ be a nonconstant
polynomial. We say that $h=0$ is an \emph{invariant algebraic
surface} of the vector field $\X$  if it satisfies $\X h=Kh$ for
some polynomial $K=K(x,y,z)\in\C[x,y,z]$, called the \emph{cofactor}
of $h=0$ \cite{LliMahVal:Lu}. Note that $K$ has degree at most $1$. The polynomial $h$
is called a \emph{Darboux polynomial}, and we also say that $K$ is
the \emph{cofactor} of the Darboux polynomial $h$. We note that a
Darboux polynomial with zero cofactor is a polynomial first
integral. Let $g,h \in \C[x,y,z]$ be coprime. We say that a
nonconstant function $e^{g/h}$ is an \emph{exponential factor} of
the vector field $\X$ given in \eqref{eq:vf} if it satisfies $\X
e^{g/h}=Le^{g/h}$ for some polynomial $L=L(x,y,z) \in \C[x,y,z]$,
called the \emph{cofactor} of $e^{h/g}$ and having degree at most
$1$. Note that this relation is equivalent to
\begin{equation}\label{ExpFac}
\big[x(y-1)-\beta z\big]\frac{\partial (g/h)}{\partial x} +
\big[\alpha(1-x^2)-\kappa y\big]\frac{\partial (g/h)}{\partial y}
+\big[x-\lambda z\big]\frac{\partial (g/h)}{\partial z} =L.
\end{equation}
For a geometric and algebraic meaning of the exponential factors see
\cite{ChrLliPer07}. A first integral $G$ of system \eqref{Nose} is
called of \emph{Darboux type} or {\it Darboux first integral} if it
is of the form
\begin{equation}\label{eq:G}
G=f_1^{\lambda_1} \cdots f_p^{\lambda_p}
\Big[\exp\Big(\frac{g_1}{h_1}\Big)\Big]^{\mu_1} \cdots
\Big[\exp\Big(\frac{g_q}{h_q}\Big)\Big]^{\mu_q},
\end{equation}
where  $f_j$ is a Darboux polynomial,
$\big[\exp\big(g_k/h_k\big)\big]^{\mu_k}$ is an exponential factor, $g_k,h_k\in\C[x,y,z]$
and $\lambda_j, \mu_k \in \C$ for $j=1,\ldots,p$, $k=1,\ldots,q$.

The Darboux theory of integrability gives a sufficient
condition for the integrability using the information about the
invariant algebraic hypersurfaces (in our case surfaces).
In practise one does not need to exploit the information about all
the invariant algebraic surfaces to construct a first integral \cite{LliMahVal:Lu}. It
is enough to find any number, say $p$, of Darboux polynomials $f_i$
with the cofactors $K_i$, and $q$ exponential factors
$\big[\exp\big(g_k/h_k\big)\big]^{\mu_k}$ with the corresponding
cofactors $L_k$ such that the linear combination
\[
\sum_{j=1}^{p} \lambda_j K_j+\sum_{k=1}^{q} \mu_j L_k=0,
\]
where $\lambda_j,\mu_k\in\C$. Then the function $G$ as in
\eqref{eq:G}  is a first integral of $\X$. For more information on
the Darboux theory of integrability see for instance
\cite{Lli04-Integ, LliZha09:DarMultInf, LliZha09:Dar:Mult} and the references therein. Note that a
polynomial or a rational first integral is a particular case of a
Darboux first integral.

For a proof of the next proposition see \cite{ChrLliPer07}.
\begin{proposition}\label{Prop:ExpFac}
The following statements hold.
\begin{itemize}
\item [(a)] If $E=e^{g/h}$ is an exponential factor for the
polynomial system \eqref{Nose} and $h$ is not a constant polynomial,
then $h=0$ is an invariant algebraic curve.
\item [(b)] Eventually $e^{g}$ can be an exponential factor, coming from the multiplicity
of the infinity.
\end{itemize}
\end{proposition}

\section{Results of the HSA dynamo equation when $\alpha \ne 0$}\label{sec.5}

\begin{proposition}\label{prop.1}
System \eqref{Nose} when either $\alpha \beta \ne 0$ or
$\alpha \kappa \ne 0$ does not admit a polynomial first integral.
\end{proposition}

\begin{proof}
Let $h$ be a polynomial first integral of system \eqref{Nose}. Then
it satisfies
\begin{equation}\label{Prop:polynomial}
\big[x(y-1)-\beta z \big]\frac{\p h}{\p x}+ \big[\alpha (1-x^2)
-\kappa y \big]\frac{\p h}{\p y}+\big[x-\lambda z \big]\frac{\p
h}{\p z}=0.
\end{equation}
Without loss of generality we can write
\begin{equation}\label{h}
h=\sum_{j=1}^n h_j(x,y,z),
\end{equation}
where each $h_j=h_j(x,y,z)$ is a homogeneous polynomial of degree
$j$ and we assume that $h_n\neq 0$ and $n \ge 1$.

Computing the terms of degree $n+1$ in \eqref{Prop:polynomial} we
get
\begin{equation}\label{eq:tampa}
xy\frac{\p h_n}{\p x}-\alpha x^2\frac{\p h_n}{\p y}=0.
\end{equation}
Thus, solving this differential equation we get $h_n=h_n[z,(\alpha
x^2+y^2)/2]$. Since $h_n\neq 0$ is a homogeneous polynomial of
degree $n\geq 1$, we conclude that
\[
h_n = c_n z^{n-2m} (\alpha x^2 +y^2)^m,
\]
for some nonnegative integer $m$ and where $c_n$ is a constant
different from zero. We introduce the notation
\begin{equation}\label{eq:notation}
\Gamma = \alpha x^2 + y^2.
\end{equation}
Then $h_n=c_n z^{n-2m} \Gamma^m$.
Now computing the terms of degree $n$ in \eqref{Prop:polynomial} we get
\begin{equation}\label{eq:tampa.1}
-(x+\beta z )\frac{\p h_n}{\p x}-\kappa y \frac{\p h_n}{\p y} +
(x-\lambda z ) \frac{\p h_n}{\p z} +xy\frac{\p h_{n-1}}{\p
x}-\alpha x^2\frac{\p h_{n-1}}{\p y}=0.
\end{equation}
Solving it with respect to $h_{n-1}$ we obtain:
\[
\begin{split}
h_{n-1} & = \pm \frac{1}{\alpha^{1/2} \Gamma} z^{-1-2m} \bigg\{
c_n \Gamma^m z^n \bigg[ \big(-(2m-n) \Gamma -2\alpha \beta m z^2 \big)
\arctan \Big(\frac{\alpha^{1/2} x}{y} \Big) \\
&\!+\!\alpha^{1/2}z\big(\!-\!2(\kappa\!-\!1) m y\!+\!\Gamma^{1/2}(2
\kappa m \!+\!\lambda(n\!-\!2m)\big)\log\Big(\frac{4(\Gamma y \!+\!
\Gamma^{1/2})}{2\kappa m \!+\! \lambda (n\!-\!2m) x \Gamma^{3/2}z}\Big)\bigg]\\
& + \alpha^{1/2} \Gamma z^{2m+1} c_{n-1}(z,\Gamma) \bigg\},
\end{split}
\]
where $c_{n-1}$ is a homogeneous polynomial in the variables $z,\Gamma$.
Since $h_{n-1}$ is a homogeneous polynomial of degree $n-1$ we
must have
\[
2m-n=0, \quad \alpha \beta m=0, \quad \alpha [2\kappa m+
\lambda(n-2m)] =0, \quad \alpha(\kappa-1) m =0.
\]
We have $n=2m$. Furthermore,
since $\alpha \beta \ne 0$ or $\alpha \kappa \ne 0$ we also have $m=0$.
This implies $n=0$, a contradiction.
\end{proof}

\begin{proposition}\label{prop.0}
Let $f$ be a Darboux polynomial of degree $n$ with nonzero cofactor
$K=\beta_0+\beta_1 x +\beta_2 y+ \beta_3 z$, $\beta_i \in \C$ for
$i=0,\ldots,3$ of system \eqref{Nose} with
 $\alpha \ne 0$. Then $\beta_1=\beta_3=0$. Furthermore, if $f_n$
is the homogeneous polynomial of degree $n \ge 1$ then
{\upshape(}see \eqref{eq:notation} for $\Gamma${\upshape)}
\begin{equation}\label{eq:fn}
f_n=c_n z^p x^{n-2m-p} \Gamma^m, \quad \beta_2=n-p-2m,
\end{equation}
where $m,p$ are some nonnegative integers and $c_n \in \C \setminus \{0\}$.
\end{proposition}

\begin{proof}
Let $f$ be a Darboux polynomial of degree $n$ with nonzero cofactor
$K=\beta_0+\beta_1 x +\beta_2 y+ \beta_3 z$, $\beta_i \in \C$ for
$i=0,\ldots,3$ of system \eqref{Nose} with
 $\alpha \ne 0$.
Then $f$ satisfies
\begin{equation}\label{Prop:Darboux}
[x(y-1)-\beta z]\frac{\p f}{\p x}+[\alpha(1-x^2)-\kappa y] \frac{\p
f}{\p y}+[x-\lambda z] \frac{\p f}{\p z}=(\beta_0 + \beta_1x +
\beta_2y + \beta_3z)h.
\end{equation}
It is easy to see by direct computations that $h$ has degree greater
than or equal to two since system \eqref{Nose} has no Darboux
polynomials of degree one with nonzero cofactor. Thus we decompose
$f$ as a sum of homogeneous polynomials similarly as in \eqref{h},
where $n\geq 1$ and $f_n\neq 0$.

Computing the terms of degree $n+1$ in \eqref{Prop:Darboux} we get
\[
x y\frac{\p f_n}{\p x}-\alpha x^2\frac{\p f_n}{\p y}=(\beta_1x +
\beta_2y + \beta_3z)f_n.
\]
Solving this linear differential equation we obtain
\[
f_n= \exp \bigg[\pm \frac{\beta_1}{\alpha^{1/2}} \arctan
\Big(\frac{\alpha^{1/2} x}{y} \Big) \bigg] \bigg(\frac{-2 (\Gamma +
y \Gamma^{1/2})}{\beta_3 xz \Gamma^{1/2}} \bigg)^{\pm \frac{\beta_3
z}{\Gamma}} x^{\beta_2} c_n(z,\Gamma),
\]
where $c_n$ is a function in the variables $z$ and $\Gamma$. Since
$f_n$ is a homogeneous polynomial of degree $n$ we must have
$\beta_1=\beta_3=0$. Furthermore
\[
f_n=c_n z^p x^{n-2m-p} \Gamma^m, \qquad \beta_2=n-p-2m,
\]
where $m,p$ are some nonnegative integers and $c_n \in \C \setminus \{0\}$.
\end{proof}

\begin{proposition}\label{prop.2}
Let $g \in \C[x,y,z]$ satisfy
\begin{equation}\label{eq:principal.bis}
\big[x(y-1)-\beta z\big]\frac{\partial g}{\partial x} +
\big[\alpha(1-x^2)-\kappa y \big]\frac{\partial g}{\partial y}
+\big(x-\lambda z\big)\frac{\partial g}{\partial z} =\beta_0+\beta_1
x+\beta_2 y+\beta_3 z,
\end{equation}
where $\beta_i\in\C$, for $i=0,1,2,3$ and not all zero and $\alpha
\ne 0$. Then
\begin{enumerate}
\item If
 $\kappa \ne 0$ or $\kappa=0$ and $\beta,\lambda \ne 0$
then the unique solution is $h=\beta_1 z$ with $L=\beta_1(x-\lambda z)$.
\item If $\kappa =0$ and $\lambda =0$
then we obtain two solutions: $h=\beta_1 z$
with $L=\beta_1 x$ and $h=-\frac{x^2} 2 +\frac y \alpha -
\frac{y^2}{2 \alpha}-\frac{\beta}2 z^2$ with $L=\beta_0(1-y)$.
\item If $\kappa =0$, $\lambda \ne 0$ and $\beta=0$
 then we obtain two solutions: $h=\beta_1 z$ with $L=\beta_1(x-\lambda z)$
and $h=-\frac{x^2} 2 +\frac y \alpha -
\frac{y^2}{2 \alpha}$ with $L=\beta_0(1-y)$.
\end{enumerate}
\end{proposition}

\begin{proof}
We first prove that $g$ is a polynomial of degree two. We proceed
by contradiction. Assume that $g$ is polynomial of degree $n \ge 3$.
We write it as a sum of its homogeneous parts as in equation~\eqref{h}
with $h_j$ replaced by $g_j$. Without loss of generality we can assume
that $g_n \ne 0$.
Then since the right-hand side of equation~\eqref{eq:principal.bis} has
degree at most one, computing the terms of degree $n+1$ in
equation~\eqref{eq:principal.bis} we get
\[
xy \frac{\partial g_n}{\partial x}
-\alpha x^2\frac{\partial g_n}{\partial y} =0,
\]
which is equation \eqref{eq:tampa} replacing $h_n$ by $g_n$.
Then the arguments
used in the proof of Proposition \ref{prop.1} imply that $n$ must be
even and that $g_n$ must be of the form $g_n=\alpha_n(\alpha x^2 +y^2)^{n/2}$
with $\alpha_n \in \C \setminus \{0\}$.

Now computing the terms in \eqref{eq:principal.bis} of degree $n \ge 3$
and
taking into account that the right-hand side of \eqref{eq:principal.bis}
has degree one, we get equation
\[
-(x+\beta z )\frac{\p g_n}{\p x}-\kappa y \frac{\p g_n}{\p y} +
(x-\lambda z ) \frac{\p g_n}{\p z} +xy\frac{\p g_{n-1}}{\p
x}-\alpha x^2\frac{\p g_{n-1}}{\p y}=0.
\]
which is equation~\eqref{eq:tampa.1} with
 $h_n$ replaced
by $g_n$ and $h_{n-1}$ replaced by $g_{n-1}$. The arguments
used in the proof of Proposition \ref{prop.1} imply that $g_n=0$. Then
we have that $g_n=0$ for $n \ge 3$, and thus, $g$ is a polynomial
of degree at most two satisfying \eqref{eq:principal.bis}. Without
loss of generality we can assume that it has no constant term. Then we
write it as
\[
g_n= g_{100}  x +g_{010} y +g_{001} z + g_{200} x^2 +g_{110} x y + g_{101}
x z +g_{020} y^2 +g_{011} y z +g_{002} z^2,
\]
where $g_{ijk} \in \C$ for $0\le i,j,k \le 2$.
Then
\[
\begin{split}
& \big[x (y-1)-\beta z\big]\big(g_{100} + 2 g_{200} x +g_{110} y + g_{101}z \big)
\\
& + \big[\alpha(1-x^2) -\kappa y \big]\big(g_{010}  + g_{110} x +
2 g_{020} y +g_{011} z \big)\\
&  + \big(x-\lambda z \big)\big(
g_{001} + g_{101} x +g_{011} y
 + 2 g_{002} z \big)\\
&  = \beta_0 + \beta_1 x +\beta_2 y+ \beta_3 z.
\end{split}
\]
Solving this equation we get:

\smallskip

\noindent \emph{Case $1$:
 $\kappa \ne 0$ or $\kappa=0$ and $\beta,\lambda \ne 0$}. In this
case
the unique solution is $h=\beta_1 z$ with cofactor $L=\beta_1(x-\lambda z)$.

\smallskip

\noindent \emph{Cases $2$: $\kappa =0$ and $\lambda =0$}
 Then we obtain two solutions: $h=\beta_1 z$
with $L=\beta_1 x$ and $h=-\frac{x^2} 2 +\frac y \alpha -
\frac{y^2}{2 \alpha}-\frac{\beta}2 z^2$ with $L=\beta_0(1-y)$.

\smallskip

\noindent \emph{Case $3$: $\kappa =0$, $\lambda \ne 0$ and $\beta=0$}.
We
 obtain two solutions: $h=\beta_1 z$ with $L=\beta_1(x-\lambda z)$
and $h=-\frac{x^2} 2 +\frac y \alpha -
\frac{y^2}{2 \alpha}$ with $L=\beta_0(1-y)$.

\smallskip

This completes the proof of the proposition.
\end{proof}

\section{Proof of Theorem \ref{thm.main.4}} \label{sec.6}

The proof of Statement a) follows directly from Proposition \ref{prop.1}.

\smallskip

To prove Statement b) we first note that if $h$ is a Darboux
polynomial with nonzero cofactor then it is easy to see by direct
computations that the unique Darboux polynomial with degree  one is
$x$. Now we assume that $h$ is an irreducible Darboux polynomial of
degree $n$ of system \eqref{Nose} with nonzero cofactor  $K=\beta_0+
\beta_1 x + \beta_2 y+ \beta_3 z$ with $\beta_i \in \C$ for
$i=0,\ldots,3$ not all zero and that has degree greater than or
equal to two. Thus we decompose $h$ as a sum of homogeneous
polynomials similarly as in \eqref{h}, where $n\geq 2$ and $h_n\neq
0$. By Proposition \ref{prop.0} we get $\beta_1=\beta_3=0$ and

\begin{equation}\label{eq:hn}
h_n=c_n z^p x^{n-p-2m} \Gamma^m,\quad c_n \in \C\setminus \{0\},
\end{equation}
and $\beta_2=n-p-2m$.

Since $h$ is irreducible  and $x=0$ is invariant, if we restrict $h$
to $x=0$ and denote it by $\bar h$ we must have that $\bar h \ne 0$. Furthermore $\bar h$ satisfies
\[
(\alpha -\kappa y) \frac{\partial \bar h}{\partial y} -\lambda
z \frac{\partial \bar h}{\partial z} =(\beta_0 +
\beta_2 y) \bar h.
\]
Hence, solving this equation we obtain
\[
\bar h = e^{-\beta_2 y/\kappa} (\kappa y-\alpha)^{-\frac{\alpha
\beta_2 +\beta_0 \kappa}{\kappa^2}}
C[z(\kappa y-\alpha )^{-\lambda/\kappa}],
\]
where $C$ is a function in the variable $
z(\kappa y-\alpha )^{-\lambda/\kappa}$. Since $\bar h$ must be a homogeneous
polynomial we must have
$\beta_2=0$, i.e., $n=p+2m$. Hence
\begin{equation}\label{eq:rosa.0}
\bar h = (\kappa y -\alpha)^{-\beta_0/\kappa}
P[z(\kappa y-\alpha)^{-\lambda/\kappa}],
\end{equation}
where $P$ is a polynomial in the variable $
z(\kappa y-\alpha )^{-\lambda/\kappa}$, and the
\[
h_n=c_n z^p (\alpha x^2 +y^2)^m,\quad c_n \in \C\setminus \{0\}.
\]

Computing the terms of degree $n$ in \eqref{Prop:Darboux} we get
\begin{equation}\label{Darboux:eq:hn1.bis}
 -x \frac{\p h_n}{\p x}-\kappa y \frac{\p h_n}{\p y} +
(x-\lambda z ) \frac{\p h_n}{\p z} +xy\frac{\p h_{n-1}}{\p
x}-\alpha x^2\frac{\p h_{n-1}}{\p y}= \beta_0 h_n.
\end{equation}
Substituting \eqref{eq:hn} into equation \eqref{Darboux:eq:hn1.bis}
and solving it with respect to $h_{n-1}$ we obtain
\[
\begin{split}
h_{n-1} & = \pm \frac{1}{\sqrt{\alpha} \Gamma z}
\bigg\{c_n \Gamma^m z^p \bigg(p \Gamma \arctan \Big(\frac{
\sqrt{\alpha} x}{y} \Big) - z\sqrt{\alpha} \big[2 (\kappa-1) m y \\
&  -
(\beta_0 +2 \kappa m + \lambda p) \sqrt{\Gamma} \log
\Big(\frac{-4(\Gamma^{1/2} + y)}{(\beta_0 +2 \kappa m+
\lambda p) x \Gamma z} \Big) \big]\bigg)\\
&  +
\alpha^{1/2} \Gamma z c_{n-1} (z,\Gamma) \bigg\},
\end{split}
\]
where $c_{n-1}$ is a function in the variables $z, \Gamma$.
Since $h_{n-1}$ is a homogeneous polynomial of degree $n-1$
and $c_n \alpha \ne 0$ we
must have
\[
p =0, \quad \beta_0 +2 \kappa m + \lambda p=0.
\]
From $p=0$ we deduce that $m=n/2$ ($n$ must be even) and $\beta_0 = -n \kappa $.
Then
\begin{equation}\label{eq:rosa.5}
h_{n-1}  = \pm c_n n(\kappa-1)(\alpha x^2 +y^2)^{n/2-1} y +
c_{n-1} z^{n-1-2 l} (\alpha x^2 +y^2)^l,
\end{equation}
for some $c_{n-1} \in \C$ and some nonnegative integer $l$. Also $h_n$ assumes now the simplified form
\[
h_n =c_n \Gamma^{n/2} =c_n (\alpha x^2 +y^2)^{n/2}, \quad c_n \in \C
\setminus \{0\}.
\]

Computing the terms of degree $n-1$ in \eqref{Prop:Darboux} we get
\[
\begin{split}
h_{n-2} & = \bigg\{2c_{n-1}(n-1-2l)z^n\Gamma^{l+2} \arctan\Big(\frac{x \alpha^{3/2}}{y}\Big)\\
&+\alpha^{1/2}4c_{n-1}lyz^{n+1} \Gamma^{1 + l} (1-\kappa)+c_n (2-n)nx^2z^{2+2l} \alpha \Gamma^{n/2}(\kappa-1)^2\\
&-\alpha^{1/2}2c_nnz^{2+2l} \Gamma^{1+n/2}(\alpha+(\kappa-1)\kappa) \log(x)\\
&-2\alpha^{1/2}c_{n-1}z^{n+1} \Gamma^{3/2+l}[(\kappa-\lambda)(n-2l)+\lambda]\log(\Theta)\bigg\}+c_{n-2}[z,\Gamma],
\end{split}
\]
where $c_{n-2}$ is a function of $z,\Gamma$ and
$\Theta=4(y+\Gamma^{1/2})/[(\kappa-\lambda)(n-2l)+\lambda].$  Since
$h_{n-2}$ is a homogeneous polynomial of degree $n-2$ we must
necessary have
\[
c_{n-1}(n-1-2l)=0.
\]
If $n-1-2l=0$, then again, since $h_{n-2}$ is polynomial we must
have $(\kappa-\lambda)(n-2l)+\lambda=0$,  which implies that
$\kappa=0$ and this is in contradiction with our assumption. On the
other hand if $c_{n-1}=0$, then calculating then
\[
h_{n-2} = -\frac{1}{2}c_n n\Gamma^{n/2-2}\Big[ (n-2) x^2 \alpha (\kappa-1)^2+2\Gamma[\alpha+\kappa(\kappa-1)] \log(x)\Big]+c_{n-2}[z,\Gamma].
\]
Thus since $h_{n-2}$ is a polynomial,  we get that
$\alpha+\kappa(\kappa-1)=0$,  which is in contradiction with our
assumption. This concludes the proof of statement b) in the theorem.

\smallskip

To prove Statement c) we note that in view of Proposition
\ref{Prop:ExpFac} if $E$ is an exponential factor of system \eqref{Nose}
with $\beta=0$, $\alpha \ne 0$ then it is of the form
\[
E= e^{g/x^n},
\]
for some nonnegative integer $n$; and $g$ and $x^n$ are coprime.
Then $g$ satisfies the equation
\[
x (y-1) \frac{\partial g}{\partial x} + \big[\alpha(1-x^2)- \kappa
y\big] \frac{\partial g}{\partial y} + (x-\lambda z) \frac{\partial
g}{\partial z} -n (y-1) g =(\beta_0+ \beta_1 x + \beta_2 y + \beta_3
z) x^n,
\]
where we have simplified by the common factor $x^n e^{g/x^n}$ and
where $\beta_i \in \C$ for $i=0,\ldots,3$. We
consider two different cases.

\smallskip

\noindent \emph{Case $1$: $n \ge 1$}. In this case if we denote
by $\bar g$ the restriction of $g$ to $x=0$ we have that
$\bar g \ne 0$ (otherwise $g$ would be divisible by $x$, which
is not possible) and $\bar g$ satisfies
\[
 (\alpha- \kappa y)
\frac{\partial \bar g}{\partial y} -\lambda z
\frac{\partial \bar g}{\partial z} =n (y-1) \bar g.
\]
Hence, $\bar g$ is a Darboux polynomial of system \eqref{Nose}
with $\beta=0$, $\kappa \ne 0$ and restricted to $x=0$. Solving
this partial differential equation we obtain
\[
\bar g = e^{-\frac{n y}{\kappa}} (\kappa y-\alpha)^{
\frac{n(\kappa-\alpha)}{\kappa^2}} C[z (\alpha-\kappa
y)^{-\lambda/\kappa}]
\]
Since $\bar g$ must be a polynomial and $n\kappa \ne 0$ we have that
$\bar g =0$, a contradiction. Hence this case is not possible.

\smallskip

\noindent \emph{Case $2$: $n=0$}. In this case $E=e^g$
where
$g$ satisfies \eqref{eq:principal.bis}. In view of
Proposition~\ref{prop.2} and since $\beta=0$ with $\kappa \ne 0$
we obtain that the unique exponential factors are
$e^z$ with cofactor $x-\lambda z$. This concludes the proof
of the statement c).

\smallskip

The proof of statement d) will be done by contradiction.
Assume that $G$ is a first integral of Darboux type.
In view
of the definition of first integral of Darboux
type in \eqref{eq:G} and taking into account statements a), b) and
c), $G$ must be of the form
\[
G= x^\lambda e^{\mu z}, \quad \text{with}
\quad \lambda, \mu \in \C.
\]
Since $G$ is a first integral it must satisfy $\X G=0$, that is,
\[
\begin{split}
\X G &= x(y-1) \frac{\partial G}{\partial  x} + \big[\alpha
(1-x^2)-\kappa y\big] \frac{\partial G}{\partial y}
+ \big(x-\lambda z \big)\frac{\partial G}{\partial z} \\
&  = \big[\lambda(y-1) +  \mu (x-\lambda z)\big] G=0.
\end{split}
\]
Hence, $\lambda(y-1) + \mu(x-\lambda z)=0$,
which implies $\lambda=\mu=0$. Then
$G=\text{constant}$, in contradiction with the fact that $G$ was
a first integral. This concludes the proof of the theorem.

\section{Proof of Theorem \ref{thm.main.5}}\label{sec.7}

The proof of Statement a) follows directly from Proposition \ref{prop.1}.

\smallskip

To prove Statement b) we first note that if $h$ is a Darboux
polynomial with nonzero cofactor then it is easy to see by direct
computations that $h$ has degree greater than or equal to two since
system \eqref{Nose} has no Darboux polynomials of degree one with
nonzero cofactor $K=\beta_0+ \beta_1 x + \beta_2 y+ \beta_3 z$ with
$\beta_i \in \C$ for $i=0,\ldots,3$ not all zero. Thus we decompose
$h$ as a sum of homogeneous polynomials similarly as in \eqref{h},
where $n\geq 2$ and $h_n\neq 0$. By Proposition \ref{prop.0} we get
$\beta_1=\beta_3=0$ and $h_n=c_n z^p x^{n-p-2m} \Gamma^m$, $c_n \in
\C\setminus \{0\}$ and $\beta_2=n-p-2m$.

Computing the terms of degree $n$ in \eqref{Prop:Darboux} we get
\begin{equation}\label{Darboux:eq:hn1}
\begin{split}
& -(x+\beta z )\frac{\p h_n}{\p x}-\kappa y \frac{\p h_n}{\p y} +
(x-\lambda z ) \frac{\p h_n}{\p z} +xy\frac{\p h_{n-1}}{\p
x}-\alpha x^2\frac{\p h_{n-1}}{\p y}\\
& = \beta_0 h_n +(n-p-2m) y h_{n-1}.
\end{split}
\end{equation}
Substituting \eqref{eq:fn} into equation \eqref{Darboux:eq:hn1}
and solving it with respect to $h_{n-1}$ we obtain
\begin{equation}\label{eq:hn1}
\begin{split}
& h_{n-1}=\pm \frac{x^{-1-2m+n-p}}{\sqrt{\alpha} \Gamma z}
\bigg\{ c_n \Gamma^m z^p \bigg( x (\alpha p x^2 + p
y^2 -2 \alpha \beta m z^2) \arctan \Big( \frac{\sqrt{\alpha} x}{y}
\Big) \\
& + \sqrt{\alpha} z \big( y (-2(\kappa-1) m x + \beta(-2m+n-p) z
\big) \\
& + \big[\beta_0 +2 (\kappa-1)m +n + (\lambda-1) p \big]
x\sqrt{\Gamma} \log \Big(\frac{-4 (\Gamma + y \Gamma^{1/2})}{(\beta_0
+2 (\kappa-1 ) m +n + (\lambda-1) p ) x \Gamma^{3/2} z } \Big) \bigg) \\
&
+ \sqrt{\alpha} x \Gamma z c_{n-1}(z,\Gamma)\bigg\}
\end{split}
\end{equation}
where $c_{n-1}(z,\Gamma)$ is a function of the variables $z,\Gamma$.
Since $h_{n-1}$ is a homogeneous polynomial of degree
$n-1$ we must have in particular that
\[
p=0, \quad \alpha \beta m=0,\quad \beta_0
+2 (\kappa-1 ) m +n + (\lambda-1) p =0.
\]

Since $\alpha \beta \ne 0$ this implies $p=m=0$ and $\beta_0=-n$. Hence
$h_{n-1}$ becomes
\[
h_{n-1} = \frac{c_n \beta n y z}{\Gamma} x^{n-1} + x^ n c_{n-1}(z,\Gamma).
\]
Again since $h_{n-1}$ must be a homogeneous polynomial of degree $n-1$ and
$\alpha \beta c_n \ne 0$ we must have $n=0$. But then $\beta_0=
\beta_1=\beta_2=\beta_3=0$ a contradiction with the fact that $h$
was a Darboux polynomial with nonzero cofactor. This concludes
the proof of statement b) in the theorem.

\smallskip

To prove statement c) we note that
from Proposition \ref{Prop:ExpFac} we can write $E =
e^g$ and $g$ satisfies \eqref{eq:principal.bis}. Then statement c)
follows now directly from Proposition \ref{prop.2}.

\smallskip

The proof of statement~d) will be done by contradiction.
Assume that $G$ is a first integral of Darboux type. We consider three
different cases:

\smallskip

\noindent \emph{Case $1$: $\kappa \ne 0$ or
$\kappa=0$ and $\lambda \beta \ne 0$}.
 Then
in view
of the definition of first integral of Darboux
type in \eqref{eq:G} and taking into account statements a), b) and
c), $G$ must be of the form
\[
G= e^{\mu z}, \quad \text{with}
\quad \mu \in \C.
\]
Since $G$ is a first integral it must satisfy $\X G=0$, that is,
\[
\begin{split}
\X G &= \big(x(y-1)-\beta z \big) \frac{\partial G}{\partial  x} +
\big(\alpha (1-x^2)-\kappa y\big) \frac{\partial G}{\partial y}
+ \big(x-\lambda z \big)\frac{\partial G}{\partial z} \\
&  =  \mu \big(x-\lambda z \big) G=0.
\end{split}
\]
Hence, $\mu(x-\lambda z)=0$,
which implies $\mu=0$. Then
$G=\text{constant}$, in contradiction with the fact that $G$ was
a first integral.

\smallskip

\noindent \emph{Case $2$: $\kappa=0$ and $\lambda =0$}. In this case
in view
of the definition of first integral of Darboux
type in \eqref{eq:G} and taking into account statements a), b) and c),
$G$ must be of the form
\[
G= e^{\mu_1 z} e^{\mu_2 (-\frac{x^2}2 +\frac y \alpha -\frac{y^2}{2 \alpha}
-\frac{\beta}2 z^2)}, \quad \text{with}
\quad \mu_1, \mu_2 \in \C.
\]
Since $G$ is a first integral it must satisfy $\X G=0$, that is,
\[
\begin{split}
\X G &= \big[x(y-1)-\beta z\big] \frac{\partial G}{\partial  x} +
\big[\alpha (1-x^2)-\kappa y\big] \frac{\partial G}{\partial y}
+ \big[x-\lambda z\big]\frac{\partial G}{\partial z} \\
&  =  \big[\mu_1 (x-\lambda z)+\mu_2 (1-y)\big] G=0.
\end{split}
\]
Hence,
$\mu_1 (x-\lambda z ) + \mu_2 (1-y) =0$,
which implies $\mu_1=\mu_2=0$. Then
$G=\text{constant}$, in contradiction with the fact that $G$ was
a first integral.

\smallskip

\noindent \emph{Case $3$: $\kappa=0$, $\lambda \ne 0$ and $\beta=0$}.
Then
in view
of the definition of first integral of Darboux
type in \eqref{eq:G} and taking into account statements a), b) and c),
$G$ must be of the form
\[
G= e^{\mu_1 z} e^{\mu_2 (-\frac{x^2}2 +\frac y \alpha -\frac{y^2}{2 \alpha})
}, \quad \text{with}
\quad \mu_1, \mu_2 \in \C.
\]
Since $G$ is a first integral it must satisfy $\X G=0$, that is,
\[
\X G  =
  \big(\mu_1 (x-\lambda z)+\mu_2 (1-y) \big) G=0.
\]
Proceeding as in the case above we reach a contradiction. This concludes the proof
of the theorem.

\section*{Acknowledgements}
The second author is partially supported by FCT through CAMGDS, Lisbon.

\end{document}